\documentclass[11pt]{article}

\usepackage[letterpaper,margin=1in]{geometry}

\usepackage{amsmath, amsthm, amssymb}
\usepackage{cite}
\usepackage{url}
\usepackage{cleveref}
\usepackage{appendix}
\usepackage{graphicx}
\usepackage{epstopdf}
\usepackage{setspace}
\usepackage{enumerate}
\usepackage{color}
\usepackage{xcolor}
\usepackage{algorithm}
\usepackage[noend]{algpseudocode}
\usepackage{multirow}
\usepackage{paralist}
\usepackage{times}

\usepackage{caption}
\usepackage{xspace}
\usepackage{mdwlist}
\usepackage{tikz}

\newtheorem{theorem}{Theorem}[section]
\newtheorem{lemma}[theorem]{Lemma}

\newtheorem{fact}[theorem]{Fact}

\newtheorem{definition}[theorem]{Definition}

\newcommand{\calvin}[1]{ [[[ \textcolor{blue}{\bf Cal:} {\em #1} ]]]}
\newcommand{\sebastian}[1]{ [[[ \textcolor{green}{\bf Seb:} {\em #1} ]]]}
\newcommand{\mohsen}[1]{ [[[ \textcolor{blue}{\bf Mohsen:} {\em #1}
  ]]]}
\newcommand{\fabian}[1]{ [[[ \textcolor{orange}{\bf Fabian:} {\em #1}
  ]]]}
\newcommand{\seth}[1]{ [[[ \textcolor{red}{\bf Seth:} {\em #1} ]]]}

\newcommand{\hide}[1]{}

\newcommand{\whp}[1][\empty]{\ensuremath{\text{w.h.p.}\ifthenelse{\equal{#1}{\empty}}{}{(#1)}}}
\newcommand{\Whp}[1][\empty]{\ensuremath{\text{W.h.p.}\ifthenelse{\equal{#1}{\empty}}{}{(#1)}}}

\algnewcommand\algorithmicswitch{\textbf{switch}}
\algnewcommand\algorithmiccase{\textbf{case}}

\algdef{SE}[SWITCH]{Switch}{EndSwitch}[1]{\algorithmicswitch\ #1\ \algorithmicdo}{\algorithmicend\ \algorithmicswitch}%
\algdef{SE}[CASE]{Case}{EndCase}[1]{\algorithmiccase\ #1}{\algorithmicend\ \algorithmiccase}%
\algtext*{EndSwitch}%
\algtext*{EndCase}%

\algnewcommand\algorithmicwithprob{\textbf{with probability}}
\algnewcommand\algorithmicotherwise{\textbf{otherwise}}

\algdef{SE}[WithProb]{WithProb}{EndWithProb}[1]{\algorithmicwithprob\ #1\ \algorithmicdo}{\algorithmicend\ \algorithmicwithprob}%
\algdef{Ce}[Otherwise]{WithProb}{Otherwise}{EndWithProb}
  {\algorithmicotherwise}%
\algtext*{EndWithProb}%
\algtext*{EndOtherwise}%

\newcommand{\FullOrShort}{short}

\ifthenelse{\equal{\FullOrShort}{full}}{
	  
  \newcommand{\fullOnly}[1]{#1}
  \newcommand{\shortOnly}[1]{}

  }{
    \usepackage[compact]{titlesec}
    
    \newcommand{\fullOnly}[1]{}
    \newcommand{\shortOnly}[1]{#1}
  }

\newcounter{tempAppCounter}
\newcounter{tempSecNumber}
\newcommand{\setAppCounter}[1]{
  \setcounter{tempAppCounter}{\arabic{theorem}}
  \setcounter{tempSecNumber}{\arabic{section}}

  \setcounter{theorem}{\arabic{ctr:#1}}
  \setcounter{section}{\arabic{ctrsec:#1}}
  \renewcommand{\thetheorem}{\arabic{section}.\arabic{theorem}}
}

\newcommand{\renewAppCounter}{
  \setcounter{theorem}{\arabic{tempAppCounter}}
  \setcounter{section}{\arabic{tempSecNumber}}
  \def\thetheorem{\oldtheorem}
}

\newcommand{\remark}[2][]
{{\color{green}\ifthenelse{\equal{#1}{}}{Remark:}{Remark (#1):} #2}}

\begin{document}
\date{}
\title{How to Discreetly Spread a Rumor in a Crowd}
\author{Mohsen Ghaffari \\ MIT \\ \texttt{ghaffari@mit.edu} \and Calvin Newport \\ Georgetown University \\ \texttt{cnewport@cs.georgetown.edu}}

\maketitle

\begin{abstract}
In this paper, we study PUSH-PULL style rumor spreading algorithms in the {\em mobile telephone model}, a variant of the classical {\em telephone model}
in which each node can participate in at most one connection per round; i.e., you can no longer have multiple
nodes pull information from the same source in a single round.
Our model also includes two new parameterized generalizations:
  (1) the network topology can undergo a bounded rate of change
(for a parameterized rate that spans from no changes to changes in every round);
and (2) in each round, each node can advertise a bounded amount of information to all of its neighbors before connection
decisions are made (for a parameterized number of bits that spans from no advertisement to large advertisements).
We prove that in the mobile telephone model with no advertisements and no topology changes,
PUSH-PULL style algorithms perform poorly with respect to a graph's vertex expansion and graph conductance
as compared to the known tight results in the classical telephone model. 
We then prove, however,
that if nodes are allowed to advertise a single bit in each round,
a natural variation of PUSH-PULL terminates in time that matches (within logarithmic factors)
this strategy's performance in the classical telephone model---even in the presence of frequent topology changes.
We also analyze how the performance of this algorithm degrades as the rate of change increases
toward the maximum possible amount.
We argue that our model matches well the properties of emerging peer-to-peer communication standards
for mobile devices, and that our efficient PUSH-PULL variation that leverages small advertisements and
adapts well to topology changes is a good
choice for rumor spreading in this increasingly important setting.
\end{abstract}

\vspace{1cm}




\section{Introduction}
\label{sec:intro}

Imagine the following scenario. Members of your organization are located throughout a crowded conference hall. 
You know a rumor that you want to spread to all the members of your organization, but you do not want anyone else
in the hall to learn it.
To maintain discreetness, 
communication
occurs only through whispered one-on-one conversations held between pairs of nearby members of your organization.
In more detail, time proceeds in rounds. In each round, each member of your organization
 can attempt to initiate a whispered conversation with a single
nearby member in the conference hall. To avoid drawing attention, each member can only whisper to one person per round.
{In this paper,} we study how quickly simple random strategies will propagate your rumor in this imagined crowded conference hall scenario.

\paragraph{The Classical Telephone Model.}
At first encounter, the above scenario seems mappable to the well-studied problem of rumor spreading  in the classical {\em telephone model.}
In more detail, the telephone model describes a network topology as a graph $G=(V,E)$ of size $n=|V|$ with a computational process (called
{\em nodes} in the following) associated with each vertex in $V$.
In this model, 
an edge $\{u,v\} \in E$ indicates that node $u$ can communicate directly with node $v$.
Time proceeds in rounds. In each round, each node can initiate a {\em connection} (e.g., place a telephone call) with a neighbor  
in $G$ through which the two nodes can then communicate. 

There exists an extensive literature on the performance of a random rumor spreading strategy called PUSH-PULL
 in the telephone model under different graph assumptions; e.g.,~\cite{chierichetti2010rumour, giakkoupis2011tight, giakkoupis2012rumor, giakkoupis2014tight}. 
The PUSH-PULL algorithm works as follows: {\em in each round, each node connects to a neighbor selected with uniform randomness;
if exactly one node in the connection is {\em informed} (knows the rumor) and one node is {\em uninformed} (does not know the rumor),
then the rumor is spread from the informed to the uninformed node.}
An interesting series of papers culminating only recently
established that PUSH-PULL terminates (with high probability)
in $\Theta((1/\alpha)\log^2{n})$ rounds in graphs with vertex expansion $\alpha$~\cite{giakkoupis2014tight},
and in $\Theta((1/\phi)\log{n})$ rounds in graphs with graph conductance $\phi$~\cite{giakkoupis2011tight}. (see Section~\ref{sec:prelim} for definitions
of $\alpha$ and $\phi$.) 

It might be tempting to use these bounds to
describe the performance of the PUSH-PULL strategy in our above conference hall scenario---{\em but they do not apply}.
A well-known quirk of the telephone model is that a given node can accept an unbounded number of incoming connections
in a single round. For example, if a node $u$ has $n-1$ neighbors initiate a connection in a given round,
in the classical telephone model $u$ is allowed to accept all $n-1$ connections and communicate with all $n-1$ neighbors in that round.
In our conference hall scenario, by contrast, we enforce the natural assumption that each node can participate in at
most one connection per round.  (To share the rumor to multiple neighbors at once might attract unwanted attention.)
The existing analyses of PUSH-PULL in the telephone model, 
which depend on the ability of nodes to accept multiple incoming connections,
do not carry over to this bounded connection setting.

\paragraph{The Mobile Telephone Model.}
In this paper, we formalize our conference hall scenario with a variant of the telephone model
 we call the {\em mobile telephone model.}
 Our new model differs from the classical version in that it now limits 
 each node to participate in at most one connection per round.
 We also introduce two new parameterized properties.
 The first is {\em stability}, which is described with an integer $\tau > 0$.
For a given $\tau$, the network topology must remain stable for intervals of at least $\tau$ rounds before changing.
The second property is {\em tag length}, which is described with an integer $b \geq 0$.
For a given $b$, at the beginning of each round, each node is allowed to publish an {\em advertisement}
containing $b$ bits that is visible to its neighbors.
Notice, for $\tau = \infty$ and $b=0$, the mobile telephone model exactly describes the
conference hall scenario that opened this paper.

Our true motivation for introducing this model, of course,
is not just to facilitate covert cavorting at conferences.
We believe it fits many emerging peer-to-peer
communication technologies better than the classical telephone model.
In particular, in the massively important space of mobile wireless devices (e.g., smartphones, tablets, networked vehicles, sensors),
 standards such as Bluetooth LE, WiFi Direct, and the Apple Multipeer Connectivity Framework,
all depend on a {\em scan-and-connect} architecture in which devices scan for nearby devices before attempting to initiate
a reliable unicast connection with a single neighbor. This architecture does not support a given device concurrently connecting with many nearby devices.
Furthermore, this scanning behavior enables the possibility of devices adding a small number of advertisement bits to their publicly
visible identifiers (as we capture with our tag length parameter), and mobility is fundamental (as we capture with our graph stability parameter).


\paragraph{Results.}
In this paper, we study rumor spreading 
in the mobile telephone model under different assumptions regarding the connectivity properties
of the graph as well as the values of model parameters $\tau$ and $b$. All upper bound results described
below hold with high probability in the network size $n$.

We begin, in Section~\ref{sec:prop}, by studying whether $\alpha$ and $\phi$ still provide useful upper bounds
on the efficiency of rumor spreading once we move from the classical to mobile telephone model.
We first prove that offline optimal rumor spreading terminates in $O((1/\alpha)\log{n})$ rounds
in the mobile telephone model in any graph with vertex expansion $\alpha$.
It follows that it is {\em possible}, from a graph theory perspective,
 for a simple distributed rumor spreading algorithm in the mobile telephone model
to match the performance of PUSH-PULL in the classical telephone model. (The question of whether simple strategies
{\em do} match this optimal bound is explored later in the paper.)
At the core of this analysis are two ideas:
(1) the size of a maximum matching bridging a set of informed and uninformed nodes at a given round describes the
maximum number of new nodes that can be informed in that round; and
(2) we can, crucially, bound the size of these matchings with respect to the vertex expansion of the graph.
We later leverage both ideas in our upper bound analysis.

We then consider graph conductance and uncover a negative answer.
In particular, we prove that offline optimal rumor spreading terminates in $O(\frac{\Delta}{\delta\cdot \phi}\log{n})$ rounds
in graphs with conductance $\phi$, maximum degree $\Delta$, and minimum degree $\delta$.
We also prove that there exist graphs where $\Omega(\frac{\Delta}{\delta\cdot \phi})$ rounds are required.
These results stand in contrast to the potentially much smaller upper bound of $O((1/\phi)\log{n})$ for PUSH-PULL in the classical
telephone model.
In other words, once we shift from the classical to mobile telephone model,
conductance no longer provides a useful upper bound on rumor spreading time.

In Section~\ref{sec:b0}, we turn our attention to studying the behavior of the PUSH-PULL algorithm in 
the mobile telephone model with $b=0$ and $\tau=\infty$.\footnote{As we detail in Section~\ref{sec:b0}, there
are several natural modifications we must make to PUSH-PULL for it to operate as intended under the new
assumptions of the mobile telephone model.}
Our goal is to determine whether this standard strategy approaches the optimal bounds from Section~\ref{sec:prop}.
For the case of vertex expansion, we provide a negative answer by constructing a graph
with constant vertex expansion $\alpha$ in which PUSH-PULL requires $\Omega(\sqrt{n})$ rounds to terminate.
Whether there exists {\em any} distributed rumor spreading algorithm that can approach optimal bounds with respect
to vertex expansion under these assumptions, however, remains an intriguing open question. 
For the case of graph conductance,
we note that a consequence of a result from~\cite{daum2015rumor} is that PUSH-PULL in this setting comes within a $\log$
factor of the (slow) $O(\frac{\Delta}{\delta\cdot \phi}\log{n})$ optimal bound proved in Section~\ref{sec:prop}.
In other words, in the mobile telephone model rumor spreading might be slow with respect to a graph's conductance, 
but PUSH-PULL matches this slow spreading time.

Finally, in Section~\ref{sec:b1},
we study PUSH-PULL in the mobile telephone model with $b=1$.
In more detail, we study the natural variant of PUSH-PULL in this setting
in which nodes use their $1$-bit tag to advertise at the beginning of each round
whether or not they are informed. We assume that informed nodes
select a neighbor in each round uniformly from the set of their uninformed neighbors (if any).
We call this variant {\em productive PUSH} (PPUSH) as nodes only attempt to 
push the rumor toward nodes that still need the rumor.

Notice, in the classical telephone model, the ability to advertise your informed status
trivializes rumor spreading as it allows nodes to implement a basic flood (uninformed nodes pull
only from informed neighbors)---which is clearly optimal. In the mobile telephone model, by contrast,
the power of $b=1$ is not obvious: a given informed node can only communicate with (at most) a single
uninformed neighbor per round, and it cannot tell in advance which such neighbor might be most
useful to inform.

Our primary result in this section,
which provides the primary upper bound contribution of this paper,
is the following: in the mobile telephone model with $b=1$ and stability
parameter $\tau \geq 1$,
PPUSH terminates
in $O((1/\alpha)\Delta^{\frac{1}{r}}r\log^3{n})$ rounds, where $r = \min\{\tau, \log{\Delta}\}$.
In other words, for $\tau \geq \log{\Delta}$, PPUSH terminates in $O((1/\alpha)\log^4{n})$ rounds,
matching (within log factors) the performance of the optimal algorithm in the mobile telephone model
{\em and} the performance of PUSH-PULL in the classical telephone model.
An interesting implication of this result is that the power gained by allowing nodes to advertise whether or not
they know the rumor outweighs the power lost by limiting nodes to a single connection per round.

As the stability of the graph decreases from $\tau = \log{\Delta}$ toward $\tau = 1$,
the performance of PPUSH is degraded by a factor of $\Delta^{1/\tau}$.
 At the core of this result is a novel analysis of randomized approximate
distributed maximal matchings in bipartite graphs,
which we combine with the results from Section~\ref{sec:prop}
to connect the approximate matchings generated by our algorithm to the graph vertex expansion.
We note that it is not {\em a priori} obvious that mobility makes rumor spreading more difficult.
It remains an open question, therefore, as to whether this $\Delta^{1/\tau}$ factor is an artifact
of our analysis or a reflection of something fundamental about changing topologies.


\paragraph{Returning to the Conference Hall.}
The PPUSH algorithm enables us to tackle the question that opens the paper: {\em What is a good way to discreetly spread a rumor in a crowd? }
One answer, we now know,
goes as follows.
If you know the rumor,
randomly choose a nearby member that does not know the rumor and attempt to whisper it
in their ear. When you do, also instruct them to make some visible sign to
indicate to their neighborhood that they are now informed; e.g., ``turn your conference badge upside down".
(This signal can be agreed upon in advance or decided by the source and spread along with the rumor.)
This simple strategy---which effectively implements PPUSH in the conference hall---will spread
the rumor fast with respect to the crowd topology's vertex expansion,
and it will do so in a way that copes elegantly and automatically to any level of encountered topology changes.
More practically speaking, we argue that in the new world of mobile peer-to-peer networking,
something like PPUSH is probably the right primitive to use to spread information
efficiently through an unknown and potentially changing network.

\section{Related Work}
\label{sec:related}
The telephone model described above was first introduced by Frieze and Grimmett~\cite{frieze1985shortest}.
A key problem in this model is {\em rumor spreading}: a rumor must spread from a single source
to the whole network.
In studying this problem, algorithmic simplicity is typically prioritized over absolute optimality. 
The PUSH algorithm (first mentioned~\cite{frieze1985shortest}), for example, simply has 
every node with the message choose a neighbor with uniform randomness and send it the message.
 The PULL algorithm (first mentioned~\cite{clarkson1992epidemic}), by contrast, has every node 
 without the message choose a neighbor with uniform randomness and 
 ask for the message. The PUSH-PULL algorithm combines those two strategies.
In a complete graph, both PUSH and PULL complete in $O(\log{n})$ rounds, with high probability---leveraging
 epidemic-style spreading behavior. Karp et~al.~\cite{karp2000randomized} proved that the
average number of connections per node when running PUSH-PULL in the complete graph is bounded at  $\Theta(\log\log{n})$.

In recent years, attention has turned toward studying the performance of PUSH-PULL with respect
to graph properties describing the connectedness or expansion characteristics of the graph.
One such measure is {\em graph conductance}, denoted $\phi$, which captures, roughly speaking,
how well-knit together is a given graph. A series of papers produced increasingly refined results with respect
to $\phi$, culminating in the 2011 work of Giakkoupis~\cite{giakkoupis2011tight} which established that
PUSH-PULL terminates in $O((1/\phi)\log{n})$ rounds with high probability in graphs with conductance $\phi$.
This bound is tight in the sense that there exist graphs with this diameter and conductance $\phi$.
Around this same time, Chierichetti et~al.~\cite{chierichetti2010rumour} motivated and initiated
the study of PUSH-PULL with respect to the graphs vertex expansion number, $\alpha$,
which measures its expansion characteristics.
Follow-up work by Giakkoupis and Sauerwald~\cite{giakkoupis2012rumor} proved that there exist
graphs with expansion $\alpha$ where $\Omega((1/\alpha)\log^2{n})$ rounds are necessary for PUSH-PULL to terminate,
and that PUSH alone achieves this time in regular graphs.
Fountoulakis et al.~\cite{fountoulakis2010rumor} proved that PUSH performs 
better---in this case, $O((1/\alpha)\log{n})$ rounds---given even stronger expansion properties.
A 2014 paper by Giakkoupis~\cite{giakkoupis2014tight} proved a matching bound of $O((1/\alpha)\log^2{n})$ for
PUSH-PULL in any graph with expansion $\alpha$.

Recent work by Daum et al.~\cite{daum2015rumor} emphasized the shortcoming of the telephone model mentioned
above: it allows a single node to accept an unlimited number of incoming connections.
They study a restricted model in which each node can only accept a single connection per round.
We emphasize that the mobile telephone model with $b=0$ and $\tau = \infty$ is equivalent to the
model of~\cite{daum2015rumor}.\footnote{There are some technicalities in this statement. A key property
of the model from~~\cite{daum2015rumor}  is how concurrent connection attempts are resolved.
They study the case where the successful connection is chosen randomly and the case where it
is chosen by an adversary. In our model,  we assume the harder case of multiple connections being resolved arbitrarily.}
This existing work proves the existence of graphs where PULL works in polylogarithmic time in the classical telephone model
but requires $\Omega(\sqrt{n})$ rounds in their bounded variation. 
 They also prove that in any graph with maximum degree $\Delta$ and minimum degree $\delta$,
 PUSH-PULL completes in  $O(\mathcal{T} \cdot \frac{\Delta}{\delta}\cdot \log{n})$ rounds, where $\mathcal{T}$ is the performance of PUSH-PULL in the classical telephone model.
 Our work picks up where~\cite{daum2015rumor} leaves off by: (1) studying the relationship between rumor spreading and 
 graph properties such as $\alpha$ and $\phi$ under the assumption of bounded connections; (2) leveraging small advertisement
 tags to identify simple strategies that close the gap with the classical telephone model results; and (3) considering the impact of 
 topology changes.
 
 Finally, from a centralized perspective, Baumann et al.~\cite{baumann2014worst} proved that in a model similar to the mobile telephone
 model with $b=1$ and $\tau = \infty$ (i.e., a model where you can only connect with a single neighbor per round
 but can learn the informed status of all neighbors in every round) there exists no
 PTAS for computing the worst-case rumor spreading time for a PUSH-PULL style strategy in a given graph. 
  
\section{Preliminaries}
\label{sec:prelim}

We will model a network topology with a connected undirected graph $G=(V,E)$.
For each $u\in V$, we use $N(u)$ to describe $u$'s neighbors and $N^+(u)$ to describe $N(u) \cup \{u\}$.
We define $\Delta = \max_{u\in V}\{| N(u)|\}$ and $\delta = \min_{u\in V}\{| N(u)|\}$.
For a given node $u\in V$, define $d(u) = |N(u)|$.
For given set $S \subseteq V$, define $vol(S) = \sum_{u\in S} d(u)$
and let $cut(S, V\setminus S)$ describe the number of edges with one endpoint in $S$ and one endpoint in $V \setminus S$.
As in~\cite{giakkoupis2011tight}, we define the {\em graph conductance } $\phi$ of a given graph $G=(V,E)$
as follows:

 \[ \phi = \min_{S \subseteq V, 0 < vol(S) \leq vol(V)/2} \frac{cut(S, V\setminus S)}{vol(S)}.\]
 
  For a given $S \subseteq V$, define the {\em boundary} of $S$, indicated $\partial S$, as follows:
 $\partial S = \{ v\in V \setminus S : N(v) \cap S \neq \emptyset\}$: that is, $\partial S$ is the set
 of nodes not in $S$ that are directly connected to $S$ by an edge.
 We define $\alpha(S) = |\partial S|/|S|$.
As in~\cite{giakkoupis2014tight}, we define the {\em vertex expansion} $\alpha$ of a given graph $G = (V,E)$
 as follows:
 
 \[  \alpha = \min_{S \subset V, 0 < |S| \leq n/2} \alpha(S). \]
 
  Notice that, despite the possibility of $\alpha(S) >1$ for some $S$, we always have $\alpha \in [0,1]$.
 Our model defined below sometimes considers a {\em dynamic graph} which can change
 from round to round. Formally, a dynamic graph ${\cal G}$ is a sequence
 of static graphs, $G_1 = (V,E_1), G_2 = (V,E_2), ...$.
 When using a dynamic graph ${\cal G}$ to describe a network topology,
 we assume the $r^{th}$ graph in the sequence describes the topology during round $r$.
 We define the vertex expansion of a dynamic graph ${\cal G}$ to be the minimum
 vertex expansion over all of ${\cal G}$'s constituent static graphs, and
the graph conductance of ${\cal G}$ to be the minimum graph conductance
 over ${\cal G}$'s static graphs. Similarly, we define the maximum and minimum degree
 of a dynamic graph to be the maximum and minimum degrees defined over all its static graphs.
 
Finally, we state a pair of well-known inequalities that will prove useful in several places below:
  
 \begin{fact}
For $p \in [0, 1]$, we have $(1-p) \leq e^{-p}$ and $(1+p) \geq 2^p$.
 \label{fact:prob}
 \end{fact}

\section{Model and Problem}
\label{sec:model}

We introduce a variation
of the classical telephone model we call the {\em mobile telephone model}.
This model describes a network topology in each round as an undirected connected graph 
$G=(V,E)$. We assume a computational process (called a {\em node}) is assigned to each vertex in $V$.
Time proceeds in synchronized rounds. 
At the beginning of each round, we assume each node $u$ knows its neighbor set $N(u)$.
Node $u$ can then select at most one node from $N(u)$ and send a connection proposal.
A node that sends a proposal cannot also receive a proposal.
However, if a node $v$ does not send a proposal, and at least one neighbor sends a proposal
to $v$, then $v$ 
can select at most one incoming proposal to accept.
(A slightly stronger variation of this model is that the accepted proposal is selected arbitrarily
by an adversarial process and not by $v$. Our algorithms work for this strong variation
and our lower bounds hold for the weaker variation.)
If node $v$ accepts a proposal from node $u$,
the two nodes are {\em connected} and can perform an unbounded amount of communication in that round.

We parameterize the mobile telephone model with two integers, $b\geq 0$
and $\tau \geq 1$. If $b>0$,
then we allow each node to select a {\em tag} containing $b$ bits to advertise at the beginning
of each round. That is, if node $u$ chooses tag $b_u$ at the beginning of a round,
all neighbors of $u$  learn $b_u$ before making their connection decisions in this round.
We also allow for the possibility of the network topology changing, which we formalize
by describing the network topology with a dynamic graph ${\cal G}$. 
We bound the allowable changes in ${\cal G}$ with a {\em stability} parameter $\tau$.
For a given $\tau$, 
${\cal G}$ must satisfy the property that we can partition it into intervals of length $\tau$,
such that all $\tau$ static graphs in each interval are the same.\footnote{Our algorithms work for many different natural notions of stability. For example,
it is sufficient to guarantee that in each such interval the graph is stable with constant probability,
or that given a constant number of such intervals, at least one contains no changes, etc.
The definition used here was selected for analytical simplicity.}
For $\tau=1$, the graph can change every round.
We use the convention of stating $\tau=\infty$ to indicate the graph never changes.

In the mobile telephone model we study the {\em rumor spreading problem}, defined as follows:
A single distinguished source begins with a {\em rumor} and the problem is solved once all nodes learn the rumor.

\section{Rumor Spreading with Respect to Graph Properties}
\label{sec:prop}

As summarized above,
a series of recent papers established that in the classical telephone model
 PUSH-PULL terminates with high probability
in $\Theta((1/\alpha)\log^2{n})$ rounds in graphs with vertex expansion $\alpha$,
and in $\Theta((1/\phi)\log{n})$ rounds in graphs with graph conductance $\phi$.
The question we investigate here is the relationship between $\alpha$ and $\phi$
and the optimal offline rumor spreading time in the mobile telephone model.
That is, we ask: once we bound connections, do $\alpha$ and $\phi$ still provide a good indicator
of how fast a rumor can spread in a graph?

\subsection{Optimal Rumor Spreading for a Given Vertex Expansion}

Our goal in this section is to prove the following property regarding optimal rumor spreading
in our model and its relationship to the graph's vertex expansion:

\begin{theorem}
Fix some connected graph $G$ with vertex expansion $\alpha$. 
The optimal rumor spreading algorithm terminates in $O((1/\alpha)\log{n})$ rounds
in $G$ in the mobile telephone model. 
\label{thm:alpha}
\end{theorem}

In other words, 
it is at least theoretically
possible to spread a rumor in the mobile telephone model
as fast (with respect to $\alpha$) as PUSH-PULL in the easier classical telephone model. 
In the analysis below, assume a fixed connected graph $G=(V,E)$ with vertex expansion $\alpha$ and $|V| =n$.

\paragraph{Connecting Maximum Matchings to Rumor Spreading.}
The core difference between our model and the classical telephone model
is that now each node can only participate in at most one connection per round.
Unlike in the classical telephone model, therefore,
the set of connections in a given round must describe a matching. 
To make this more concrete, we first define some notation.
In particular, given some $S\subset V$,
let $B(S)$ be the bipartite graph with bipartitions $(S,V \setminus S)$
and the edge set $E_S = \{ (u,v): (u,v) \in E$, $u\in S$, and $v\in V \setminus S\}$.
Also recall that the {\em edge independence number} of a graph $H$,
denoted $\nu(H)$, describes the maximum matching on $H$.
We can now formalize our above claim as follows:

\begin{lemma}
Fix some $S \subset V$. The maximum number of concurrent connections
between nodes in $S$ and $V\setminus S$ in a single round is $\nu(B(S))$.
\label{lem:match}
\end{lemma}

We can connect the smallest such maximum matchings in our graph $G$
to the optimal rumor spreading time. 
Our proof of the following lemma
combines the connection between matchings and rumor spreading captured in Lemma~\ref{lem:match},
with the same high-level analysis structure deployed in  existing studies of rumor
spreading and vertex expansion in the classical telephone model (e.g.,~\cite{giakkoupis2014tight}):

\begin{lemma}
Let $\gamma = \min_{S\subset V, |S| \leq n/2}\{ \nu(B(S))/|S|  \}$.
It follows that optimal rumor spreading in $G$ terminates in $O((1/\gamma)\log{n})$ rounds.
\label{lem:gamma}
\end{lemma}
\begin{proof} 
Assume some subset $S\subset V$ know the rumor.
Combining Lemma~\ref{lem:match} with the definition of $\gamma$,
it follows that: (1) if $|S| \leq n/2$, then at least $\gamma|S|$ new nodes can learn the rumor in the next round;
and (2) if $|S| \geq n/2$, then at least $\gamma |V\setminus S|$ new nodes can learn the rumor.

So long as Case $1$ holds, the number of informed nodes grows by at least a factor of $(1+\gamma)$
in each round. By Fact~\ref{fact:prob}, after $t$ rounds, the number of informed nodes has
grown to at least $1\cdot(1 + \gamma)^t \geq 2^{\gamma \cdot t}$. Therefore, after at most $t_1 = (1/\gamma)\log{(n/2)}$
rounds, the set of informed nodes is of size at least $n/2$.

At this point we can start applying Case $2$ to the shrinking set of uninformed nodes. 
Again by Fact~\ref{fact:prob}, after $t_2$ additional rounds, the number of uninformed nodes has
reduced to at most $(n/2)\cdot(1 - \gamma)^t \leq e^{- \gamma \cdot t}$. Therefore, after at most $t_2 = \Theta((1/\gamma)\ln{n})$
rounds, the set of uninformed nodes is reduced to a constant. After this point, a constant number of additional rounds 
is sufficient to complete rumor spreading. 
It follows that $t_1 + t_2 = O((1/\gamma)\log{n})$ rounds is enough to solve the problem.
\end{proof}

\paragraph{Connecting Maximum Matching Sizes to Vertex Expansion.}
Given Lemma~\ref{lem:gamma},
to connect rumor spreading time to vertex expansion in our mobile telephone model,
it is sufficient to bound maximum matching sizes with respect to $\alpha$.
In particular, we will now argue that $\gamma \geq \alpha/4$ (the details of this constant factor do not matter much; $4$ happened
to be convenient for the below argument).
Theorem~\ref{thm:alpha} follows directly from the below result combined with
Lemma~\ref{lem:gamma}.

\begin{lemma}
Let $\gamma = \min_{S\subset V, |S| \leq n/2}\{ \nu(B(S))/|S|  \}$.
It follows that $\gamma \geq \alpha/4$.
\label{lem:msize}
\end{lemma}
\begin{proof}
We can restate the lemma equivalently as follows:
{\em for every $S \subset V$, $|S| \leq n/2$, the maximum matching on $B(S)$
is of size at least $(\alpha |S|)/4$.}
We will prove this equivalent formulation.

To start, fix some arbitrary subset $S\subset V$ such that $|S| \leq n/2$.
Let $m$ be the size of a maximum matching on $B(S)$.
Recall that $\alpha \leq \alpha(S) = |\partial S|/|S|$.
Therefore, if we could show that $|\partial S| \leq 4m$, we would be done.
Unfortunately, it is easy to show that this is not always the case. Consider a partition $S$ in which a single node $u\in S$
is connected to large number of nodes in $V \setminus S$, and these are the only edges leaving $S$.
The vertex expansion in this example is large while the maximum matching size is only $1$ (as all nodes in $\partial S$
share $u$ as an endpoint).
To overcome this problem, we will, in some instances,
instead consider a related smaller partition $S'$ such that $\alpha(S') \geq \alpha$
is small enough to ensure our needed property.
In more detail, we consider two cases regarding the size of $m$:

{\em The first case} is that $m \geq |S|/2$. By definition, $\alpha \leq 1$.
It follows that $m \geq (|S|\alpha)/2$, which more than satisfies our claim.

{\em The second (and more interesting) case} is that $m  < |S|/2$. Let $M$ be a maximum matching
of size $m$ for $B(S)$. 
Let $M_S$ be the endpoints in $M$ in $S$.
We define a smaller partition $S' = S \setminus M_S$.
Note, by the case assumption, $|S'| \geq |S|/2$.
We now argue that every node in $\partial S'$ is also in $M$.
To see why, assume for contradiction that there exists some $v\in \partial S'$ that is not in $M$.
Because $v\in \partial S'$, there must exist some edge $(u,v)$, where $u\in S'$.
Notice, however, because $u$ is in $S'$ it is not in $M$.
If follows that we could have added $(u,v)$ to our matching $M$ defined on $B(S)$---contradicting
the assumption that $M$ is maximum.
We have established, therefore, that $|\partial S'| \leq 2m$.
It follows:

\[ \alpha \leq \alpha(S') \leq 2m/|S'| = 2m/(|S| - m) \stackrel{(m < |S|/2)}{<} \frac{2m}{|S|/2} = (4m)/|S|, \]

\noindent from which it follows that $\alpha |S| < 4m \Rightarrow m > (\alpha |S|)/4$, as needed to satisfy the claim.
\end{proof}

\subsection{Optimal Rumor Spreading for a Given Graph Conductance}

In the classical telephone model PUSH-PULL terminates in $O((1/\phi)\log{n})$ rounds in a graph
with conductance $\phi$. Here we prove optimal rumor spreading might be much slower in the mobile telephone model.
To establish the intuition for this result, consider a star graph with one center node and $n-1$ points.
It is straightforward to verify that the conductance of this graph is constant.
But it is also easy to verify that at most one point can learn the rumor per round in the mobile telephone
model, due to the restriction that each node (including the center of the star) can only participate in
one connection per round. In this case, every rumor spreading algorithm will be a factor
of $\Omega(n/\log{n})$ slower than PUSH-PULL in the classical telephone model.

Below we formalize a fine-grained version of this result, parameterized with 
maximum and minimum degree of the graph. We then leverage Theorem~\ref{thm:alpha},
and a useful property from~\cite{giakkoupis2014tight}, to prove the result tight.

\begin{theorem}
Fix some integers $\delta, \Delta$, such that $1 \leq \delta \leq \Delta$.
There exists a graph $G$ with minimum degree $\delta$ and maximum degree $\Delta$, such that every rumor spreading algorithm
requires $\Omega(\Delta/(\delta\cdot \phi))$ rounds in the mobile telephone model.
In addition, for every graph with minimum degree $\delta$ and maximum degree
$\Delta$, the optimal rumor spreading algorithm terminates in $O(\Delta/(\delta\cdot \phi)\cdot \log{n})$ rounds
in the mobile telephone model.
\label{thm:phi}
\end{theorem}
\begin{proof} 
Fix some $\delta$ and $\Delta$ as specified by the theorem statement.
Consider a generalization of the star, $S_{\delta, \Delta}$, where the center is composed of a clique containing $\delta$ nodes,
and there are $\Delta$ point nodes, each of which is connected to (and only to) all $\delta$ nodes in the center clique.
We first establish that the graph conductance of $S_{\delta, \Delta}$ is constant.
To see why, we consider three cases for each set $S$ considered in the definition of $\phi$.
The first case is that $S$ includes only $k$ center nodes. 
Here it follows:

\[\frac{cut(S, V\setminus S)}{vol(S)} = \frac{k(\delta-k) + k\Delta}{ k(\delta-1)+k\Delta}  \geq \frac{k(\delta-k) + k\Delta}{ 2k\Delta} \geq 1/2.\]

\noindent Next consider the case where $S$ contains only $k$ point nodes. 
Here it follows: 

\[\frac{cut(S, V\setminus S)}{vol(S)} = \frac{k\delta}{k\delta} = 1 .\]

\noindent Finally, consider the case where $S$ contains $\ell$ points nodes and $k$ center nodes. 
An important observation is that $vol(S) \leq vol(V)/2$. It follows that $k \leq \delta/2$: because, given
that every edge is symmetrically adjacent to the center, it would otherwise follow:
$vol(S) > vol(V)/2$. We now lower bound the conductance by $1/4$ as follows: If $\ell \leq \Delta/2$, then for each node in $S$, at least half of its neighbors are outside $S$, which shows that $\frac{cut(S, V\setminus S)}{vol(S)} \geq 1/2$. On the other hand, suppose $\ell > \Delta/2$. Then at least $1/4$ of the edges between the center clique and point nodes go out of $S$, because $\ell > \Delta/2$ and $k \leq \delta/2$. And at least $1/2$ of the edges inside the clique go out of $S$, because $k \leq \delta/2$. Since the conductance is simply a weighted average of these two ratios, we have $\frac{cut(S, V\setminus S)}{vol(S)} \geq 1/4$.

Having established the conductance is constant, to conclude the lower bound component of the theorem proof, it is sufficient to note that at most
$\delta$ of the $\Delta$ points can learn the rumor per round. It follows that every rumor spreading algorithm
requires at least $\Delta/\delta$ rounds.

Finally, to prove that $O(\Delta/(\delta\cdot \phi)\cdot \log{n})$ rounds is always sufficient for a graph with minimum and maximum
degrees $\delta$ and $\Delta$, respectively, we leverage the following property (noted in~\cite{giakkoupis2014tight}, among other places):
for every graph $G$ with vertex expansion $\alpha$ and graph conductance $\phi$, $(\delta/\Delta)\phi \leq \alpha$,
which directly implies $\alpha \geq (\delta/\Delta)\phi$.
Combining this observation with Theorem~\ref{thm:alpha},
the claimed upper bound follows.
\end{proof}

\section{PUSH-PULL with $b=0$}
\label{sec:b0}

We now study the performance of PUSH-PULL in the mobile telephone model
with $b=0$ and $\tau=\infty$. We investigate its performance with respect to
the optimal rumor spreading performance bounds from
Section~\ref{sec:prop}.
In more detail, we consider the following natural variation of PUSH-PULL,
adapted to our model:

\begin{quote}
In even rounds, nodes that know the rumor choose a neighbor at random and attempt to establish a connection
to PUSH the message. In odd rounds, nodes that do not know the rumor choose a neighbor at random and
attempt to establish a connection to PULL the message. 
\end{quote}

\noindent We study this PUSH-PULL variant with respect to both graph conductance and vertex expansion.

\paragraph{Graph Conductance Analysis.}
We begin by considering the performance of this algorithm with respect to graph conductance.
Theorem~\ref{thm:phi} tells us that for any minimum and maximum degree $\delta$ and $\Delta$, respectively,
the optimal rumor spreading algorithm completes in $O(\Delta/(\delta\cdot \phi)\cdot \log{n})$ rounds,
and there are graphs where $\Omega(\Delta/\delta)$ rounds are necessary.
Interestingly, as noted in Section~\ref{sec:related},
Daum et~al.~\cite{daum2015rumor} proved that the above algorithm 
terminates in $O( \mathcal{T} \cdot \frac{\Delta}{\delta} \cdot \log{n} )$ rounds,
where $\mathcal{T}$ is the optimal performance of PUSH-PULL in the classical telephone model.
Because $\mathcal{T} \in O((1/\phi)\log{n})$ in the classical setting, 
the above algorithm should terminate in $O( \frac{\Delta}{\delta \cdot \phi} \cdot \log^2{n} )$ rounds in our model---nearly
matching the bound from Theorem~\ref{thm:phi}.
Put another way, rumor spreading potentially performs poorly with respect to graph conductance,
but PUSH-PULL with $b=0$ nearly matches this poor performance.

Notice, we are omitting from consideration here the mobile telephone model with graphs that can
change (non-infinite $\tau$). The analysis from~\cite{daum2015rumor} does not hold in this case and new work would be required
to bound PUSH-PULL in this setting with less stable graphs.
By contrast, when studying uniform rumor spreading below for $b=1$, 
we explicitly include the graph stability as a parameter in our time complexity. 

\paragraph{Vertex Conductance Analysis.}
Arguably, the more important optimal time complexity bound to match is the $O((1/\alpha)\log{n})$ bound established in Theorem~\ref{thm:alpha},
as it is similar to the performance of PUSH-PULL in the telephone model.
We show, however, that for $b=0$, the algorithm can deviate from the optimal performance of Theorem~\ref{thm:alpha}
by a factor in $\Omega(\sqrt{n})$. This observation motivates our subsequent study of the $b=1$ case where we prove
that uniform rumor spreading can nearly match optimal performance with respect to vertex expansion.

\begin{lemma}\label{lem:badGraph} There is a graph $G$ with constant vertex expansion, in which the above algorithm would need at least $\Omega(\sqrt{n})$ rounds to spread the rumor, with high probability.
\end{lemma}
\begin{proof}
We start with describing the graph. The graph has two sides, the left side $L$ is a complete graph with $n/2$ nodes, and the right side $R$ is an independent set of size $n/2$. The connection between $L$ and $R$ is made of two edge-sets: (1) a matching of size $n/2$, (2) a full bipartite graph connecting $R$ to a subset $L^*$ of $L$ where $L^*=\lfloor{\sqrt{n}\lfloor}$. See \Cref{fig:BadGraph-cropped}. Note that this graph has constant vertex expansion.

We argue that, w.h.p., per round at most $O(\sqrt{n})$ new nodes of $R$ get informed (regardless of the current state). Hence, the spreading takes $\Omega(\sqrt{n})$ rounds, despite the good constant vertex expansion of $G$. 

\begin{figure}[t]
	\centering
		\includegraphics[width=0.60\textwidth]{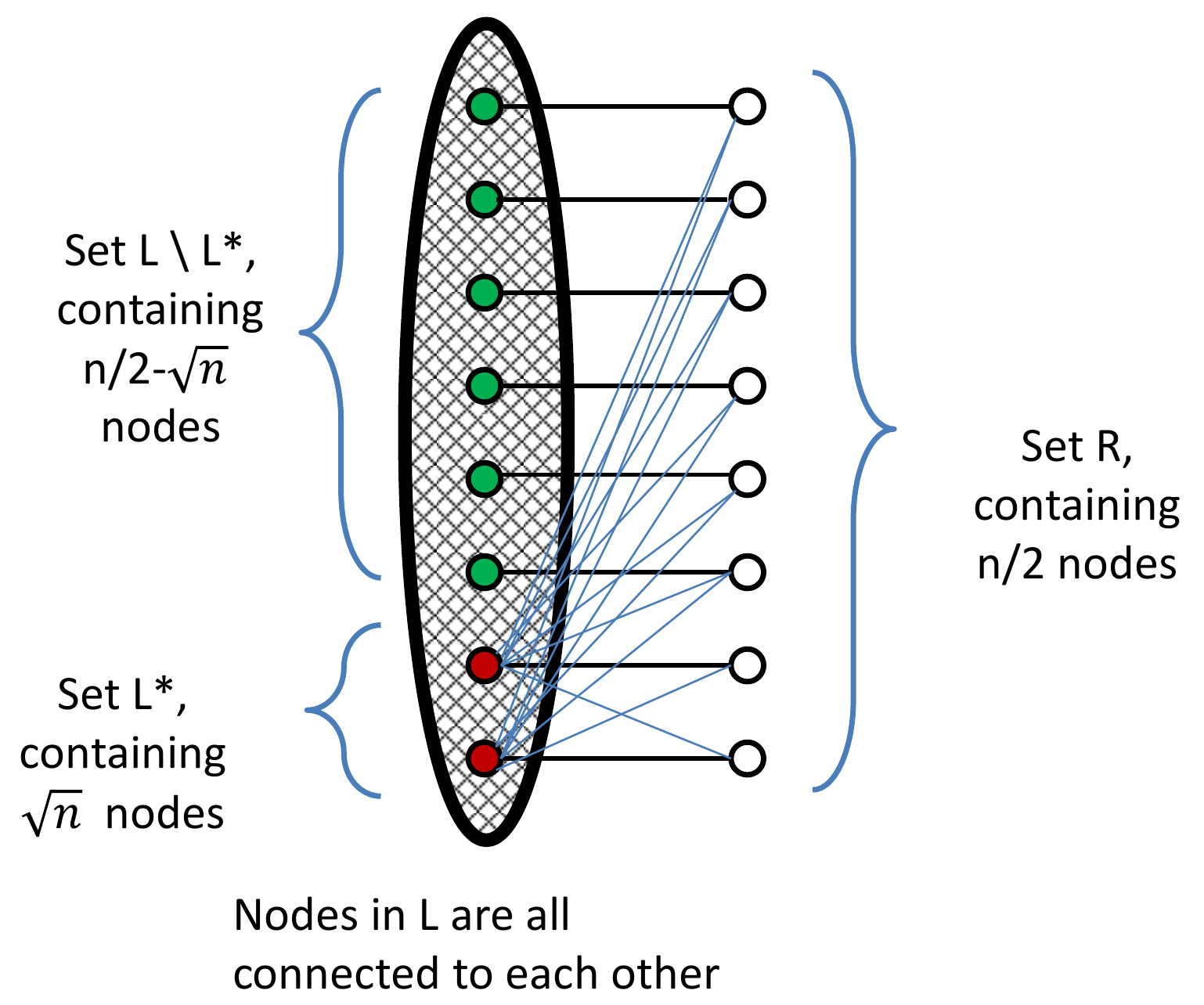}
	\caption{A graph with constant vertex expansion where {PUSH}-{RPULL} takes $\Omega(\sqrt{n})$ rounds. The nodes on the left side $L$ form a complete graph, and the nodes on the right side $R$ are only connected to $L$, via a matching of size $n/2$ and a complete bipartite graph to subset $L^* \subset L$ of size $\sqrt{n}$.}
	\label{fig:BadGraph-cropped}
\end{figure}

First, let us consider the PUSH process: each node in $L\setminus L^*$ pushes to a node in $R$ with probability $\frac{1}{|L|-1} = \frac{1}{n/2-1}$. Hence, over all nodes of $L\setminus L^*$, we expect $O(1)$ pushes to $R$, which, by Chernoff, means we will not have more than $O(\log n)$ such pushes, with high probability. On the other hand, each node in  $L^*$ pushes to at most one node in $R$. Hence, the total number of successful pushes to $R$ is at most $\sqrt{n}+O(\log n)$, with high probability. 

Now, we consider the restricted pull process (RPULL): for each node $v\in R$, with probability 1/$|L^*|$, the pull lands in $L\setminus L*$. Hence, overall, we expect $\frac{n}{2} \cdot \frac{1}{\sqrt{n}} \leq \sqrt{n}/2 $ pulls to land in $L\setminus L^*$. Hence, with high probability, at most $O(\sqrt{n})$ nodes of $R$ get informed by pulling nodes of $L \setminus L^*$. On the other hand, the vast majority of the pulls of $R$-nodes lands in $L^*$ but due to the restriction in the RPULL, nodes of $L^*$ can only respond to $L^*$ of these pulls, which is $\sqrt{n}$ many. Hence, the number of $R$-nodes informed via pulls is at most $O(\sqrt{n})$, with high probability.

Taking both processes into account, we get that per round, $O(\sqrt{n})$ new nodes of $R$ get informed, which means rumor spreading will require at least $\Omega(\sqrt{n})$ rounds. 
\end{proof}

\section{PUSH-PULL with $b=1$}
\label{sec:b1}

In the previous section, we proved that PUSH-PULL in the mobile telephone model and $b=0$ fails to match 
the optimal vertex expansion
bound by a factor in $\Omega(\sqrt{n})$ in the worst case.
 Motivated by this shortcoming, we turn our attention to the setting where $b=1$.
In particular, we consider the following natural variant of PUSH-PULL adapted to our model with $b=1$.
We call this algorithm {\em productive PUSH} (or, PPUSH) as nodes leverage the $1$-bit tag
to advertise their informed status and therefore keep connections productive.\footnote{We drop the PULL behavior
form PUSH-PULL in this algorithm description as it does not help the analysis.  Focusing just on the PUSH
behavior simplifies the algorithm even further. }

\begin{quote}
At the beginning of each round, each node uses a single bit to advertise whether or not it is {\em informed} (knows the rumor).
Each informed node that has at least one uninformed neighbor, chooses an uninformed neighbor with uniform randomness
and tries to form a connection to send it the rumor.
\end{quote}


\noindent We now analyze PPUSH in a connected network $G$ with vertex expansion $\alpha$
and stability factor $\tau \geq 1$.  Our goal is to prove the following theorem:

\begin{theorem}
Fix a dynamic network $G$ of size $n$ with vertex expansion $\alpha$ and stability factor at least $\tau$, $1 \leq \tau \leq \log{\Delta}$.
The PPUSH algorithm solves rumor spreading in $G$ in
$(1/\alpha)\Delta^{1/\tau} \tau \log^3{n}$ rounds, with high probability in $n$.
\label{thm:ppush}
\end{theorem}
 
 To prove this theorem, the core technical part is in studying the success
 of PPUSH over a stable period of $\tau$ rounds, which we present in Section~\ref{sec:upper:1}. This analysis bounds the number of new nodes that receive the message in the stable period with respect to the size of the maximum
 matching defined over the informed and uninformed partitions at the beginning of the stable period.
 In Section~\ref{sec:upper:2}, we connect this analysis back to the vertex expansion of the graph
 (leveraging our earlier analysis from Section~\ref{sec:prop} connecting $\alpha$ to edge independence numbers),
 and carry it through over multiple stable periods until we can show rumor spreading completes.
 
\subsection{Matching Analysis}
\label{sec:upper:1}

Our main theorem in this section lower bounds the number of rumors that spread across
a bipartite subgraph of the network over $r$ stable rounds.

\begin{theorem}
Fix a bipartite graph $G$ with bipartitions $L$ and $R$, such
$|R| \geq |L| = m$ and  $G$ has a matching of size $m$. 
Assume $G$ is a subgraph of some (potentially) larger network $G'$,
and all uninformed neighbors in $G'$ of nodes in $L$ are also in $R$.
Fix an integer $r$, $1 \leq r \leq \log{\Delta}$, where $\Delta$ is the maximum degree of $G$.
Consider an $r$ round execution of PPUSH in $G'$ in which the nodes in $L$ start
with the rumor and the nodes in $R$ do not.
With constant probability: at least $\Omega(\frac{m\Delta^{-1/r}}{r\log{n}})$ nodes in $R$ learn the rumor.
\label{thm:matching}
\end{theorem}

We start with some helpful notation. 
For any $L' \subseteq L$ and $R' \subseteq R$,
let $G(L',R')$ be the subgraph of $G$ induced
by the nodes $L'$ and $R'$.
Similarly, let $N_{L',R'}$ and $deg_{L',R'}$ be the neighbor and degree functions,
respectively,
defined over $G(L',R')$.

We begin with the special case $r=1$. We then move to our main analysis which handles all $r\geq 2$.
(Notice, our result below $r=1$ provides an approximation of $O(\sqrt{\Delta \log \Delta})$ which is 
tighter than the $\Delta$ approximation for this case claimed by Theorem~\ref{thm:matching}.
We could refine the theorem claim to more tightly capture
performance for small $r$, but we leave it in the looser more general form for the sake of concision in the result statement.)
%

\begin{lemma} \label{lemma-r1} For $r=1$, PPUSH produces a matching of size $\Omega(m/\sqrt{\Delta \log \Delta})$, with constant probability.
\end{lemma}
\begin{proof}
Consider the maximal matching $M$ in the bipartite graph, and let $L_M$ be the set of informed endpoints of this matching. Divide nodes of $L_M$ into $\log \Delta$ classes based on their degree in the bipartite graph $G$, by putting nodes of degree in $[2^{i-1}, 2^{i}]$ in the $i^{th}$ class. By pigeonhole principle, one of these classes contains at least $|M|/\log \Delta = m/\log \Delta$ nodes. Let $L^{i}_{M}$ be this class and let $d=2^{i-1}$ be such that the degrees in this large class are in $[d, 2d]$. Now each node $v$ in $L_M$ pushes to its pair in the matching $M$ with probability $1/deg(v) \geq 1/(2d)$. Hence, we expect $\frac{m}{2d\log \Delta}$ uninformed endpoints of $M$ to be informed by getting a push from their matching pair. If $d = O(\sqrt{\Delta/\log \Delta})$, Chernoff bound already shows us that with high probability in this expectation and thus with at least constant probability, the matching size is at least $m \cdot \Omega(1/\sqrt{\Delta \log \Delta})$, hence establishing the Lemma's claim. Suppose on the contrary that $d = \Omega(\sqrt{\Delta/\log \Delta})$.

Now, let $R^*$ be the set of $R$-nodes adjacent to $L^{i}_M$. Call each node $v \in R^*$ high-degree if $deg(v) \geq d$. 
Each $L^{i}_M$-node pushes to each of its adjacent neighbors with probability at least $1/(2d)$, which means each high-degree node in $R^*$ gets at least one push with probability at least $1-(1-\frac{1}{2d})^{2d} > 1/4$. 
Note that there are at least $md/\log\Delta$ edges of $G$ incident on $L^{i}_M$. Hence, the number of edges incident on $R^*$ is also at least $md/\log\Delta$. Now either at least $md/(2\log\Delta)$ edges are incident on high-degree nodes of $R^*$, or at least $md/(2\log\Delta)$ edges are incident on low-degree nodes of $R^*$. In the former case, since each high-degree node has degree at most $\Delta$, there must be at least $md/(2\Delta \log \Delta) = \Omega(m/\sqrt{\Delta\log \Delta})$ high degree nodes. Since each of these gets hit with probability at least $1/4$, we expect at least $\Omega(m/\sqrt{\Delta\log \Delta})$ such hits. Due to the negative correlation of these hits, the Lemma's claim follows from Chernoff bound. 

Suppose on the contrary that we are in the latter case and $md/(2\log\Delta)$ edges are incident on low-degree nodes. Each low-degree node $v \in R^*$ gets hit with probability at least $1-(1-\frac{1}{2d})^{deg(v)} \geq \Theta(deg(v)/d)$. Since summation of degrees among low-degree nodes is at least $md/(2\log\Delta)$, we get that the expected number of hit low-degree nodes is at least $\Omega(m/\log \Delta) \gg \Omega (m/\sqrt{\Delta \log \Delta})$. A Chernoff bound concentration then completes the proof.          
\end{proof}

%

\noindent We start now the proof of the $r\geq 2$ case by making a claim that says if for a given large subset of $L$ that
has a relatively small degree sum,
a couple rounds of the algorithm run on this subset will either generate a large enough matching,
or leave behind a subset with an even smaller degree sum.

\begin{lemma}
Fix any $i \in [r]$,  $L' \subseteq L$, and $R' \subseteq R$, such that: 
$|R'| \geq |L'|  \geq m/16$; 
$\sum_{u\in L'} deg_{L',R'}(u) \leq m\Delta^{1-\frac{i-1}{r}}$; all uninformed neighbors of $L'$ in $G'$ are in $R'$;
and $G(L',R')$ has a matching of size $|L'|$.
With high probability in $n$, one of the following two events will occur if we execute
PPUSH with the nodes in $L'$ knowing the rumor and the nodes in $R'$ not knowing the rumor:

\begin{enumerate}

\item within two rounds, at least $\Omega(\frac{m\Delta^{-1/r}}{r\log{n}})$ nodes in $R'$ learn the rumor; or
\item after one round, we can identify subsets $L''\subseteq L'$, $R'' \subseteq R'$,
with $R''$ containing only nodes that do not know the rumor, such that
$|R''| \geq |L''| \geq (1-1/r)^2\cdot |L'|$;   $\sum_{u\in L''} deg_{L'',R''}(u) \leq m\Delta^{1-\frac{i}{r}}$;
$R''$ contains all uninformed neighbors of $L''$ nodes in $G'$;
and $G(L'', R'')$ has a matching of size $|L''|$.
\end{enumerate}
\label{lem:matching1}
\end{lemma}
\begin{proof}
This proof consists of several steps, which we present one by one below.
Before doing so, let $M$ be a maximum matching
on $G(L',R')$. Because we assume $|L'| = |M|$,
we know every node in $L'$ shows up in $M$.
For every node in $u$ in $M$, we use term {\em $u$'s original match},
to refer to the node $u$ is matched with in $M$.
Every node in $L'$ has a match and so do at $|M|$ nodes in $R'$.
In addition, when describing the behavior of the PPUSH algorithm,
we describe the event in which an informed node $u$ attempts
to connect to an uninformed neighbor $v$ as $u$ {\em sending a proposal}
to $v$. A given node $v$, therefore, might have many proposal sent to it in a given round.

\paragraph{Step \#1: Remove High-Degree Nodes in $L'$.}
The lemma assumptions tell us that the sum of degrees of nodes in $L'$ is not too large.
A natural implication is that at most a small fraction of these nodes can have a degree larger than the smallest
possible average degree given by this bound. 

In more detail, fix $\delta_i = \Delta^{1-\frac{i-1}{r}}\cdot r \cdot 16$.
We claim that at most a $1/r$ fraction of nodes in $L'$ have a degree of size at least $\delta_i$.
To see why this claim is true we first note that by the lemma assumptions, $|L'| \geq m/16$.
It follows that if more than $(1/r)\cdot|L'|$ nodes have a degree of size at least $\delta_i$,
then the total sum of degrees would be greater than $(1/r)\cdot|L'| \cdot  \delta_i \geq m\cdot \Delta^{1-\frac{i-1}{r}}$,
which contradicts the degree bound assumption.

Let $L_1' \subseteq L'$ be the subset of $L'$ with degrees upper bounded by $\delta_i$.
Notice, as just established, $L_1'$ has at least a $(1-1/r)$-fraction of the nodes from $L'$.
Let $R_1'$ be the subset of $R'$ that remains after we remove from $R'$ 
any node that is no longer connected to $L_1'$.
Notice, for every node $u\in L_1'$, $u$'s original
match is still in $R_1'$. Therefore, 
we get $|R_1'| \geq |L_1'|$, and the property that there is matching of size $|L_1'|$
in $G(R_1',L_1')$.
Finally, because we only removed nodes from $R'$ if they were not connected to
$L_1'$, we know every node in $L_1'$ still has all of its uninformed neighbors from $R'$ included in $R_1'$.

\paragraph{Step \#2: Remove High-Degree Nodes in $R_1'$.}
We now run one round of the PPUSH algorithm in our graph
and bound its behavior on the nodes in $G(L_1',R_1')$.
To do so, we first fix $X_v$, for each $v\in R_1'$,
to be the random variable describing the number of proposals received by $v$ from nodes in $L_1'$ in this round.
We are interested in the nodes with large expected values for $X$.
In particular, we fix $H = \{ v\in R_1' : E[X_v] \geq c\log{n} \}$,
for some constant $c\geq 1$ that we fix later.
Notice, for each $v\in H$, we can define $X_v = \sum_{u\in N_{L_1',R_1'}(v)} Y_{u,v}$,
where for each $u\in N_{L_1',R_1'}(v)$, the variable $Y_{u,v}$ is a $0/1$ indicator variable indicating that $u$ sent a proposal to $v$.

Notice, for $u \neq u'$, $Y_{u,v}$ and $Y_{u',v}$ are independent.
Therefore, $X_v$ is defined as the sum of independent random variables.
It follows that we can apply a Chernoff bound to achieve concentration
on the mean $\mu = E[X_v]$.
A straightforward consequence of this concentration is that the probability $X_v \geq 1$ is at least $1-n^{-h}$,
where $h \geq 1$ is a constant that grows with the constant $c$ selected for our definition of $H$.
Notice, this probability only {\em lower bounds} the probability that $v$ receives a proposal,
as it only focuses on proposals arriving from nodes in $L_1'$. 
Other neighbors of $v$ not in $L_1'$ might also send $v$ a proposal, which would only
{\em increase} the probability that $v$ receives a proposal---helping our goal
in this step of establishing that these high degree nodes receive proposals with high probability.
We now consider all nodes in $H$.
There might be dependencies between different $X$ variables,
but we can bypass these dependencies by applying a simple union bound to
 establish that the probability that {\em every} node in $H$
receives at least one proposal is still high with respect to $n$ (for sufficiently large $c$ in our definition of $H$).

Moving forward, we assume this event occurs. 
We now want to remove from consideration some nodes from our bipartite subgraph.
In particular, for every node $v\in R_1'$ that receives a proposal in this round,
we remove $v$ from $R_1'$.  Notice, we do not require that the proposal
came from $L_1'$ to remove $v$. That is, if $v$ receives a proposal from
any node---be it in $L_1'$ or not---we remove it.
In addition, for each $v$ we remove, if $v$ is the original match of some node $u$
in $L_1'$, we also remove $u$ from $L_1'$.
Let $R_2'$ and $L_2'$ be the nodes that remain from $R_1'$ and $L_1'$, respectively.
Every node in $L_2'$ still has its original match in $R_2'$,
therefore $|R_2'| \geq |L_2'|$,
and $G(L_2', R_2')$ has a maximum matching of size $|L_2'|$.
Also, since we only removed nodes from $R_1'$ that received the rumor,
every node in $L_2'$ still has all of its uninformed neighbors remaining in $R_2'$.


\paragraph{Step \#3: Check the Number of Nodes Left in $G(L_2',R_2')$.} 
We now consider how many nodes were removed from consideration in the last
step. 
The first case we consider is that these matches moved
more than a $1/r$ fraction of the nodes from $L_1'$.
If this occurs, then it follows that at least $|L_1'|/r$ nodes learned
the rumor in this one round. 
In the first step, however, we established that $|L_1'| \geq (m/16)(1-1/r) \in \Omega(m)$.
Therefore the number of nodes that learn the rumor is in $\Omega(m/r)$, which is which is more than large enough to directly satisfy the bound of bound of $\frac{m\Delta^{-1/r}}{r\log{n}}$ claimed in objective 1 of the lemma. 

Moving forward, therefore,
we assume that no more than a $1/r$ fraction of the nodes in $L_1'$ were
removed to define $L_2'$.
Combining the reductions from the first two steps, we have $|L_2'| \geq |L'|\cdot(1-1/r)^2$.


	\paragraph{Step \#4: Bound the Degree Sum of $L_2'$.}
	Notice, if we set $L'' = L_2'$ and $R'' = R_2'$, then 
	we have shown so far that these sets satisfy
	{\em most} of the conditions required by objective $2$ of the lemma statement. 
	Indeed, the only condition we have not yet analyzed is the sum of the degrees of the nodes in $L_2'$, which we examine next. 
	
	We divide the possibilities for this sum into two cases.
	The first case is that $\sum_{u\in L_2'} deg_{L_2',R_2'}(u) \leq m\Delta^{1-\frac{i}{r}}$.
	That is, that the sum is small enough to satisfy objective $2$ of the lemma.
	If this occurs, we are done.
	The second case is that $\sum_{u\in L_2'} deg_{L_2',R_2'}(u) > m\Delta^{1-\frac{i}{r}}$.
	We will now show that if this is true then, with high probability in $\Delta$, in one additional
	round of the algorithm we inform enough new nodes to satisfy objective $1$ of the lemma.
	
	To do so, first recall that by definition, every node in $L_2'$ has a degree at most $\delta_i$,
	and therefore any given edge in $G(L_2', R_2')$ is selected by a $L_2'$ node
	with probability at least $1/\delta_i$.
	For each $v\in R_2'$, let $X_v$ be the expected number of proposals that $v$ receives
	from nodes in $L_2'$ during this round.
	Using our above lower bound on edge selection probability,
	we can calculate: $E[X_v] \geq deg_{L_2', R_2'}(v)/\delta_i.$
	Let $Y = \sum_{v\in R_2'} X_v$ be the total number of proposals received by $R_2'$ nodes from $L_2'$ nodes.
	By linearity of expectation: 
	
	\[ E[Y] \geq \sum_{v\in R_2'} deg_{L_2',R_2'}(v)/\delta_i= \sum_{u\in L_2'} deg_{L_2',R_2'}(u)/\delta_i> m\Delta^{1-\frac{i}{r}} \cdot (1/\delta_i) = (m\Delta^{-\frac{1}{r}})/ \Theta(r). \]
	
	\noindent As defined, $Y$ is {\em not} necessarily the sum of independent random variables,
	as there could be dependencies between different $X$ values.
	However, it is straightforward to verify that for any $u \neq v$,  $X_u$ and $X_v$ are {\em negatively associated}: 
	$u$ receiving more proposal can only reduce the number of proposals received by $v$. 
	Because we can apply a Chernoff bound to negatively associated random variables,
	we can achieve concentration around the expected value for $Y$. Note that $E[Y] \geq c\log{n}$ as otherwise the claim of the lemma reduces to informing just one node which holds trivially. 
	It follows that with high probability in $n$, we have $Y \geq (m\Delta^{-\frac{1}{r}})/ \Theta(r)$.

        We are not yet done. Recall that $Y$ describes the total number of {\em proposals} received by
        nodes in $R_2'$, not the total number of {\em nodes} in $R_2'$ that receive proposals.
        It is, however, this later quantity that we care about.
                Fortunately, at this step we can leverage Step \#2,
                during which, with high probability,
                we removed all nodes in $R_1'$ that expected to receive at least $\log{n}$
                proposals.
                In other words, w.h.p., for every $v\in R_2'$: $E[X_v] \in O(\log{n})$.
        
       For any $v\in R_2'$, we can bound the probability that $X_v > c\log{n}$, for some sufficiently large constant $c \geq 1$,
       to be polynomially small in $n$ (with an exponent that increases with $c$).
       By a union bound, the probability that any node in $R_2'$ receives more than $c\log{n}$ values is still
       polynomially small in $n$. Assume, therefore, that this $c\log{n}$ upper bound holds for all nodes.
       It follows that if $Y \geq (m\Delta^{-\frac{1}{r}})/ \Theta(r)$, the  number of unique nodes receiving
       proposals is at least $\frac{Y}{\Theta(\log n)} \geq \frac{m\Delta^{-\frac{1}{r}}}{\Theta(r \log n)} = \Omega\big(\frac{m\Delta^{-\frac{1}{r}}}{r \log n}\big)$.

       We can simple combine the high probability bounds
       on the size of $Y$ being large and the size of $X_v$ being small (for every relevant $v$), by a simple union bound on all of those events, as each holds with high probability in $n$ and we certainly have at most $n+1$ such events.

       Pulling together the pieces, we have shown that 
       if the degree sum on $L_2'$ is too large, 
       then with high probability in $n$ at least $\frac{m\Delta^{- 1/r}}{r\log{n}}$ new nodes are informed in this round---satisfying the lemma.
       
     To conclude, we note that the above analysis only applies to the number of nodes in $R_2'$ that are informed by nodes in $L_2'$.
     It is, of course, possible that some nodes in $R_2'$ are also informed by nodes outside of $L_2'$. 
     This behavior can only help this step of the proof as we are proving a lower bound on the number of informed
     nodes and this can only increase the actual value.
 %
	%
\end{proof}

We now leverage Lemma~\ref{lem:matching1} to prove Theorem~\ref{thm:matching}. 
The following argument establishes a base case that satisfies the lemma preconditions
of Lemma~\ref{lem:matching1} and then repeatedly applies it $r$ times.
Either: (1) a matching of sufficient size is generated along the way (i.e., case $1$ of
the lemma statement applies); or (2) we begin round $r$ with a set $L'$ with
size in $\Omega(m)$ that has an average degree in $\Theta(\Delta^{1/r})$---in which case it is
easy to show that in the final round we get a matching of size $\Omega(\frac{m\Delta^{-1/r}}{r\log n})$.

\begin{proof}[Proof (of Theorem~\ref{thm:matching}).]
Fix a bipartite graph $G$ with bipartitions $L$ and $R$ with a matching of size $|L| = m$, and a value $r$,
as specified by the theorem statement preconditions. 
If $r=1$, the claim follows directly from \Cref{lemma-r1}. Assume in the following, therefore, that  $r\geq 2$.

We claim that we can apply Lemma~\ref{lem:matching1} to $L' = L$, $R' = R$, and $i=1$.
To see why, notice that this definition of $L'$ satisfies the preconditions
$L' \subseteq L$, $R' \subseteq R$, and $|R'| \geq |L'| \geq m/16$.
It also satisfies the condition requiring all of the uninformed neighbors of $L'$ to be in $R'$.
Finally, because we fixed $i=1$, it holds that: $\sum_{u\in L'} deg_{L',R'}(u) \leq m\Delta^{1-\frac{i-1}{r}} = m\Delta$,
as there are $m$ nodes in $L'$ each with a maximum degree of $\Delta$.

Consider this first application of Lemma~\ref{lem:matching1}.
It tells us that, w.h.p., either we finish after one or two rounds,
or after a single round we identify a smaller bipartitate graph $G(L'',R'')$,
where $L''$ and $R''$ satisfy all the properties needed to apply the Lemma to $L' = L''$, $R' = R''$,
and $i=2$.
We can keep applying this lemma inductively, each time increasing the value of $i$,
until either: (1) we get through $i=r-1$; (2) an earlier application of the lemma generates
a sufficiently large matching to satisfy the theorem; or (3) at some point before
either option 1 or 2, the lemma fails to hold. Since the third possibility happens with probability polynomially small in $n$ at each application, we can use a union bound and conclude that with high probability, it does not happen in any of the iterations. Ignoring this negligible probability, we focus on the other two possibilities. 

Before that, let us discuss a small nuance in applying the lemma $r$ times. We need to ensure that the specified $L'$ sets are always of size at least $m/16$, as required to keep applying the lemma. Notice, however,
that we start with an $L'$ set of size $m$,
and the lemma guarantees it decreases by a factor of at most $(1-1/r)^2$.
Therefore, after $i < r \leq  \log{\Delta}$ applications,
$|L''| \geq (1-1/r)^{2i}\cdot m > (1/4)^{2i/r}\cdot m > (1/4)^2\cdot m = m/16$.

Going back to the two possibilities, if option 2 holds, we are done. On the other hand, if option 1 holds,  we have one final step in our argument. In this case, we end up with having identified a bipartite subgraph $G(L'',R'')$ with a maximum matching of size at least $|L''| \geq m/16$. 
We also know $\sum_{u\in L''} deg_{L'',R''}(u) \leq m\Delta^{1-\frac{i}{r}} = m \Delta^{1/r}$ as $i=r-1$. 
In this case, it holds trivially that at most $m/32$ nodes $u$ of $L''$ have $deg_{L'',R''}(u) \geq 32 \Delta^{1/r}$. Hence, at least $m/32$ nodes $u$ have degree at most $32\Delta^{1/r}$. Now, each of these proposes to its own match in $R''$ with probability at least $\Delta^{-1/r}/32$. Thus, we expect $\Theta(m\Delta^{-1/r})$ nodes of $R''$ to receive proposals directly from their matches. Note that these events are independent. Moreover, we have $m\Delta^{-1/r} = \Omega(\log n)$ as otherwise the claim of the theorem would be trivial. Therefore, w.h.p., $\Theta(m\Delta^{-1/r})$ nodes of $R''$ receive proposals from their pairs. Hence, at least $\Theta(m\Delta^{-1/r})$ nodes of $R''$ get informed, thus completing the proof.
\end{proof}

\subsection{Connecting Rumor Spreading to Bipartite Matchings and Proving \Cref{thm:ppush}}
\label{sec:upper:2}

We now leverage the matching analysis from Section~\ref{sec:upper:1} to bound the rumor spread over time, hence eventually proving \Cref{thm:ppush}. 

\paragraph{Preliminaries.}
Divide the rounds into {\em stable phases} each consisting of $\tau$ rounds, such that the graph
does not change during a stable phase.
We label these phases $1,2,...$.
Let $V$ be the node set for $G$.
Let $S_t \subset V$, for some phase $t \geq 1$,
be the subset of {\em informed} nodes that know rumor at the beginning of phase $t$.
Let $U_t = V \setminus S_t$, where $V$ is the node set of $G$.
Let $f(\tau) = \tau \Delta^{ 1/\tau } \log{n}$ be the approximation factor on the maximum matching
provided by Theorem~\ref{thm:matching}.
We define the notion of a {\em good} phase with respect to this approximation factor:

\begin{definition}
Fix some phase $t$. 
If $|S_t| \leq n/2$, we call this phase {\em good} if $|S_{t+1}| \geq \big(1+ \frac{ \alpha }{ 4 \cdot f(\tau)} \big) S_t$.
Else if $|S_t| > n/2$, we call this phase {\em good} if $|U_{t+1}| \leq \big(1- \frac{ \alpha }{ 4\cdot f(\tau)}\big) |U_t|$.
\end{definition}

Notice, the factor of $4$ in the above definition comes from Lemma~\ref{lem:msize},
from our earlier analysis connecting the size of a maximum matching
across any partition to the vertex expansion of the graph.
 
\paragraph{Bounding the Needed Good Phases.}
We next bound the number of good phases needed to complete rumor spreading.
Intuitively, until the rumor spreads to at least half the nodes,
each good phase increases the number of informed nodes by a fractional factor of at least $\gamma = \alpha/(4f(\tau))$.
Given that we start with $1$ informed node, after $t$ such increases, the number of informed nodes is at 
least $1\cdot (1+\gamma)^t \geq 2^{t\gamma}$ (by Fact~\ref{fact:prob}).
Therefore, we need $t \geq (1/\gamma)\log{(n/2)}$ good phases to get the rumor to at least half the nodes.

Once we have informed half the nodes, we can flip our perspective. We now decrease our uninformed nodes
by a factor of at least $\gamma$. 
If we start with no more than $n/2$ uninformed nodes,
then after $t$ good phases, the number of uninformed nodes left is no more than $(n/2)(1-\gamma)^t < (n/2)e^{-\gamma t}$ (also by Fact~\ref{fact:prob}).
Similar to before, $t \geq (1/\gamma)\ln{(n/2)}$ good phases is sufficient to reduce these remaining
nodes down to a constant number at which we can complete the rumor spreading.
(See the proof of Lemma~\ref{lem:gamma} for the details of this style of argument.)
We capture this intuition formally as follows:

\begin{lemma}
After $t_{max} \in O(  (f(\tau)/\alpha) \cdot \ln{n} )$ good phases, PPUSH has solved rumor spreading.
\label{lem:good}
\end{lemma}

\paragraph{Bounding the Fraction of Phases that are Good.}
Our final step is to bound how many phases are needed before we have achieved, with high probability,  the number of good phases
required by Lemma~\ref{lem:good} to solve rumor spreading.

We begin by focusing on the probability that a given phase is good.
It is in this analysis that we pull together many of the threads woven so far throughout this paper.
In particular, we consider the partition between informed and uniformed nodes.
The maximum matching between these partitions describe the maximum number of
new nodes that might be informed. We can leverage Theorem~\ref{thm:matching}
to prove that with at least constant probability, we spread rumors
to at least a $1/f(\tau)$-fraction of this matching.
We then leverage our earlier analysis of the relationship between matchings and vertex expansion,
to show that this matching generates a factor of $\alpha/4$ increase in informed nodes (or decrease in uninformed,
depending on what stage we are in the analysis). This matches our definition of {\em good}.
Formally:

\begin{lemma}
There is some constant probability $p$
such that each phase is good with probability at least $p$.
\label{lem:good:2}
\end{lemma}
\begin{proof}
Consider the maximum matching $M$ between the partitions 
$S_t$ and $U_t$.
Let $m=|M|$.
A direct implication of Lemma~\ref{lem:msize} (see the equivalent formulation in the first line of the proof),
is the following:

\begin{itemize}

\item if $|S_t| \leq |U_t|$ then $m \geq |S_t|(\alpha/4)$;
\item else if $|U_t| > |S_t|$ then $m \geq |U_t|(\alpha/4)$.

\end{itemize}

\noindent We now apply Theorem~\ref{thm:matching} to bound what fraction of this matching we can expect to inform in the $\tau$
rounds of the phase that follows.
In more detail, set $L$ to be the $m$ nodes in $M$ from $S_t$, 
set $R$ to be the uninformed neighbors of nodes in $L$, and set $r=\tau$.
It is easy to verify that these values satisfy the preconditions of Theorem~\ref{thm:matching}.
The theorem tells us that with constant probability, at least $m/f(r) = m/f(\tau)$ new nodes
learn the message in this phase.
Combined with our case analysis from above, it follows that if
if $|S_t| \geq |U_t|$, then this is at least $|S_t|(\alpha/(f(\tau)4))$ new nodes,
and if $|U_t| > |S_t|$ then this is at least $m \geq |U_t|(\alpha/(f(\tau)4))$ new nodes.
In both cases, we have satisfied the definition of {\em good} with constant probability.
\end{proof}

\noindent We know from Lemma~\ref{lem:good:2} that each phase is good with constant probability.
We know from Lemma~\ref{lem:good}, that $t_{max}$ good phases are sufficient. 
We must now combine these two observations to determine how many phases are needed to generate
$t_{max}$ good phases with high probability.

A starting point in this analysis is to define a random indicator variable $X_t$ for each phase $t$ such that:

\[  
X_t = \begin{cases}
   1 & \text{if phase $t$ is good}\\
   0 & \text{else.}
   \end{cases}
\]

Let $Y_T = \sum_{t=1}^T X_t$ be the number of good phases out of the first $T$ phases.
By linearity of expectation and Lemma~\ref{lem:good:2}: $E(Y_t) = \Theta(T)$.
Combining this observation with Lemma~\ref{lem:good} it follows
that the {\em expected time} to solve rumor spreading with PPUSH is in $O(t_{max})$.

We are seeking, however, a high probability bound to prove Theorem~\ref{thm:ppush}.
We cannot simply apply a Chernoff bound to concentrate around $E(Y_T)$, as
for $t\neq t'$, $X_t$ and $X_{t'}$ are not necessarily independent. 
Our final theorem proof will leverage a stochastic dominance argument to overcome
this obstacle.

\begin{proof}[Proof (of Theorem~\ref{thm:ppush}).]
According to Lemma~\ref{lem:good:2}, 
there exists some constant probability $p$ that lower bounds,
for every phase, the probability that the phase is good.
For each $t$, we define a trivial random variable $\hat X_t$ that is $1$ with independent probability $p$,
and otherwise $0$.
By definition, for each phase $t$, regardless of the history through phase $t-1$,
$X_t$ stochastically dominates $\hat X_t$.
It follows that if $\hat Y_T = \sum_{t=1}^{T} \hat X_t$ is greater than $x$ with some probability $\hat p$,
then $Y_T$ is greater than $x$ with probability at least $\hat p$.
A Chernoff bound applied to $\hat Y_T$, for $T=c\cdot t_{max}$ (where $c\geq 1$ is a sufficiently
large constant define with respect to the constant from Lemma~\ref{lem:good:2} and the Chernoff form,
and $t_{max}$ is provided Lemma~\ref{lem:good}), provides
that $\hat Y_T$ is at least $t_{max}$ with high probability in $n$.
It follows the same holds for $Y_T$.
By Lemma~\ref{lem:good}, this is a sufficient number of good phases to solve rumor spreading.
To obtain the final round bound we first note that the upper bound $T$ on phases simplifies as:
\[ T = O(t_{max}) = O( (f(\tau)/\alpha) \ln{n} ) =    (1/\alpha)\tau \Delta^{1/\tau} \log^2{n}. \]

\noindent We further note that each phase has at most $\tau \leq \log{\Delta}$ rounds,
and that $n\geq \Delta, m$, from which it follows the total number of required rounds
is upper bounded by $O((1/\alpha)\tau \Delta^{1/\tau}\log^3{n})$, as needed.

\end{proof}







 


\end{document}